\documentclass[sigconf]{aamas}  

\usepackage{booktabs}

\usepackage{romannum}

\usepackage{algpseudocode,algorithm,algorithmicx}

\algrenewcommand\algorithmicrequire{\quad \;\;\textbf{Input:}}
\algrenewcommand\algorithmicensure{\quad \; \,\textbf{Output:}}

\usepackage{mathtools}
\newcommand\numberthis{\addtocounter{equation}{1}\tag{\theequation}}

\newtheorem{assumption}{Assumption}

\newcommand{\ignore}[1]{}

\DeclareMathOperator*{\argmin}{arg\,min}
\newcommand{\norm}[1]{\|#1\|}

\begin{document}

\title{Spatial-Temporal Moving Target Defense: A Markov Stackelberg Game Model }  



\author{Henger Li, Wen Shen, Zizhan Zheng}
\affiliation{
\institution{Department of Computer Science, Tulane University}
\address{ 6823 Saint Charles Ave, New Orleans, LA 70118}
}
\email{{hli30, wshen9, zzheng3}@tulane.edu}

\begin{abstract}  
Moving target defense has emerged as a critical paradigm of protecting a vulnerable system against persistent and stealthy attacks. To protect a system,  a defender proactively changes the system configurations to limit the exposure of security vulnerabilities to potential attackers. In doing so, the defender creates asymmetric uncertainty and complexity for the attackers, making it much harder for them to compromise the system. In practice, the defender incurs a switching cost for each migration of the system configurations. The switching cost usually depends on both the current configuration and the following configuration. Besides,  different system configurations typically require a different amount of time for an attacker to exploit and attack. Therefore, a defender must simultaneously decide both the optimal sequences of system configurations and the optimal timing for switching. In this paper, we propose a Markov Stackelberg Game framework to precisely characterize the defender's spatial and temporal decision-making in the face of advanced attackers. We introduce a relative value iteration algorithm that computes the defender's optimal moving target defense strategies. Empirical evaluation on real-world problems demonstrates the advantages of the Markov Stackelberg game model for spatial-temporal moving target defense.
\end{abstract}


\maketitle


\section{Introduction}


Moving target defense (MTD) has established itself as a powerful framework to counter persistent and stealthy threats that are frequently observed in Web applications~\cite{vikram2013nomad,taguinod2015toward}, cloud-based services~\cite{jia2014catch,peng2014moving},  database systems~\cite{zhuang2012simulation}, and operating systems~\cite{hong2015assessing}. The core idea of MTD is that a defender 
proactively switches configurations (a.k.a., attack surfaces) of a vulnerable system to increase the uncertainty and complexity for potential attackers and limits the resources (e.g., window of vulnerabilities) available to them~\cite{jajodia2011moving}. In contrast to MTD techniques, traditional passive defenses typically use analysis tools to identify vulnerabilities and detect attacks. 

A major factor for a defender to adopt MTD techniques rather than passive defensive applications is that adaptive and sophisticated attackers often have asymmetric advantages of resources (e.g., time, prior knowledge of the vulnerabilities) over the defender due to the static nature of system configurations~\cite{sood2012targeted,okhravi2015moving}. For instance, clever attackers are likely to exploit the system over time and subsequently identify the optimal targets to compromise  without being detected by the static defensive applications~\cite{zhuang2014towards}. In such case, a well-crafted MTD technique is more effective because it changes the static nature of the system and increases the attackers' complexity and cost of mounting successful attacks ~\cite{evans2011effectiveness}. In this way, it reduces or even eliminates the attackers' asymmetric advantages of the resources. 

Despite the promising prospects of MTD techniques, it is challenging to implement optimal MTD strategies.  Two factors contribute to the difficult situation. On one hand, the defender must make the switching strategies sufficiently unpredictable because otherwise the attacker can thwart the defense.  On the other hand, the defender should not conduct too frequent system migrations because each migration incurs a switching cost that depends on both the current configuration and the following configuration. Thus, the defender must make a careful tradeoff between  the effectiveness and the cost efficiency of the MTD techniques. The tradeoff requires the defender to simultaneously decide both the optimal sequences of system configurations (i.e., the next configuration to be switched to) and the optimal timing for switching. To this end, an MTD model must precisely characterizes the defender's spatial-temporal decision-making.

A natural approach to model the strategic interactions between the defender and the attacker is to use game-theoretic models such as the zero-sum dynamic game~\cite{zhu2013game} and the Stackelberg Security game (SSG) model~\cite{vadlamudi2016moving,sengupta2017game,sengupta2018moving}. The dynamic game model is general enough to capture different transition methods of the system states and various information structures but optimal solutions to the game are often difficult to compute. This model also assumes that the switching cost from one configuration to another is fixed~\cite{zhu2013game}. In the SSG model, the defender commits to a mixed strategy that is independent of the state transitions of the system~\cite{vadlamudi2016moving,sengupta2017game}. The attacker adopts a best-response strategy to the defender's mixed strategy. While optimal solutions to an SSG can be obtained efficiently, the SSG models neglect the fact that both the defender's and the attacker's strategy can be contingent on the states of the system configurations.  To address this problem, \citeauthor{feng2017stackelberg} incorporate both Markov decision process and Stackelberg game into the modeling of the MTD game~\cite{feng2017stackelberg}. 

Many MTD models~\cite{sengupta2017game,sengupta2018moving,chowdhary2018mtd,jajodia2018share} do not explicitly introduce the concept of the defending period, although they usually assume that the defender chooses to execute a migration after a constant time period~\cite{sengupta2019survey}. A primary reason is that when both time and the system's state influence the defender's decision making, optimal solutions to the defender's MTD problem are non-trivial~\cite{manadhata2013game,sengupta2019survey}. Recently, \citeauthor{li2019optimal} has proposed to incorporate timing into the defender's decision making processes~\cite{li2019optimal}. Their work assumes that there is a positive transition probability between any two configurations (so that the corresponding Markov decision process is unichain), which may lead to sub-optimal MTD strategies.  
Further, their model assumes that all the attackers have the same type, which might not be true in reality. To solve the defender's optimal MTD problem in more general settings, an MTD model that precisely models the strategic interactions between the defender and the attacker is in urgent need.

\noindent{\bf Our contributions.} In this paper, we propose a general framework called the {\em Markov Stackelberg Game} (MSG) model for spatial-temporal moving target defense. The MSG model enables the defender to implement optimal defense strategy that is contingent on both the source states and the destination states of the system. It also allows the defender to simultaneously decide which state the system should be migrated to and when it should be migrated. In the MSG model, we formulate the defender's optimization problem as an average-cost semi-Markov decision process (SMDP)~\cite{puterman2014markov} problem.  We present a 
value iteration algorithm that can solve the average-cost SMDP problem efficiently after transforming the original average-cost SMDP problem into a discrete time Markov decision process (DTMDP) problem. We empirically evaluate our algorithm using real-world data obtained from the National Vulnerability Database (NVD)~\cite{o2009national}.  Experimental results demonstrate the advantages of using MSG over the state-of-the-art approaches in MTD.

\section{Related Work}
Our MSG model is built on an important game-theoretic model called {Stackelberg games}~\cite{von2010market} that has broad applications in many fields, including spectrum allocation~\cite{zhang2009stackelberg}, smart grid~\cite{yu2015real} and security~\cite{sinha2018stackelberg}. In a Stackelberg game, there are two players: a leader and a follower. The leader acts first by committing to a strategy. The follower observes the leader's strategy and then maximizes his reward based on his observation~\cite{an2016stackelberg}.  
Each player in a Stackelberg game has a set of possible {\em pure strategies} (i.e., actions) that can be executed. A player can also play a mixed strategy that is a distribution over the pure strategies. 

 Stackelberg games have been extensively used to model the strategic interactions between a defender and an attacker in security domains~\cite{sinha2018stackelberg}.  This category of Stackelberg games is often called {\em Stackelberg security games} (SSGs).  In SSGs, a defender aims to defend a set of targets using a limited number of resources. An attacker learns the defender's strategy and mounts attacks to maximize his benefits.  The pure strategy for a defender is an allocation of the defender's limited resources to the targets, while a mixed strategy is a probability distribution over all the possible allocations. The defender's optimization problem in an SSG is to compute a mixed strategy that maximizes her expected utility (i.e., minimizes her expected losses) given that the attacker learns the defender's mixed strategy and performs a best response to the defender's action.  Whenever there is a tie,  the attacker always breaks the tie in favor of the defender~\cite{nguyen2016towards}. This solution concept is called the {\em strong Stacklberg equilibrium}~\cite{breton1988sequential}. In our paper, we consider the strong Stacklberg equilibrium as the solution concept for the MSG game.

\section{Model}

In this section, we describe the Markov Stackelberg Game (MSG) model for moving target defense where a defender and an attacker compete for taking control of the system.  Specifically, we first introduce key notations used to model the system configuration, the defender and the attacker. We then present both the attacker's and the defender's optimization problems in the spatial-temporal moving target defense game. Finally, we compare our MSG model with the Stackelberg Security Game model~\cite{paruchuri2008playing,sinha2018stackelberg} that has been extensive studied in security research.

\subsection{System Configuration}
In moving target defense, a defender proactively shifts variables of a computing system to increase uncertainty for the attacker to mount successful attacks. The variables of the system are often called adaptive aspects~\cite{zhuang2014towards,zhuang2015theory}. Typical adaptable aspects include IP address, port number, network protocol, operating system, programming language, machine replacement, and memory to cache mapping~\cite{zhuang2015theory}. A computing system usually have multiple adaptable aspects. Let $D \in \mathbb{N}$ denote the number of adaptive aspects for the computing system and $\Phi_{i}$ the set of sub configurations in the $i$th adaptive aspect. 
If the defender selects the sub configuration $\phi_{i} \in \Phi_{i}$ for the $i$th adaptive aspect, 
then the configuration state of the system is denoted as $s = (\phi_{1}, \phi_{2}, \dots, \phi_{D})$.  Here, $s$ is a generic element of the set of system configurations $S = \Phi_{1} \times \Phi_{2} \times \dots \times \Phi_{D}$. Let $n = |S|$ denote the number of system configurations in $S$.

\begin{example}
Consider a mobile app with two adaptive aspects: the programming language aspect $\Phi_{1} = \{Python, Java, JavaScript\}$ and the operating system aspect $\Phi_{2} =  \{iOS, Android\}$. If the defender selects $Java$ for the programming language aspect and $iOS$ for the operating system aspect, then the composite system configuration is $s = (\phi_{1}, \phi_{2}) = (Java, iOS)$.  The maximum number of valid configurations for the app is 6. However, it is likely that only a subset of  the configuration set is valid due to both physical constraints and performance requirement of the system. 
\end{example}

\subsection{The Attacker}
We consider a stealthy and persistent attacker that constantly exploits vulnerabilities of the system configurations to gain benefits.
\subsubsection{Attacker Type}
A successful attack requires a competent attacker that has the expertise to exploit the vulnerabilities. For instance, to gain control of a remote computer, the attacker needs the capability to obtain the control of at least one of the following resources: the IP address of the computer, the port number of a specific application running on the machine, the operating system type, the vulnerabilities of the application, or the root privilege of the computer~\cite{zhuang2015theory}. 
The attacker's ability profile is called the {\em attacker type}. There are different attacker types. Let $L$ denote the set of attacker type. Each attacker type $l \in L$ is capable of mounting a set of attacks $A_{l}$. The attacker type space $A_{l}$ is a nonempty set of the attack methods that each targets one vulnerability in one adaption aspect, which may affect multiple sub configurations in that aspect.
Multiple attacks may target the same vulnerability but may have different benefit (loss) to the attacker (defender) as defined below.
Whether an attack method belongs to an attacker type space or not depends on the application scenarios.
Figure~\ref{fig:relation} illustrates the relationship between attacker types, attack methods, adaptation aspects, sub-configuration parameters 
and system configurations.

\begin{figure}[t]
\centering
 \includegraphics[width=.85\linewidth]{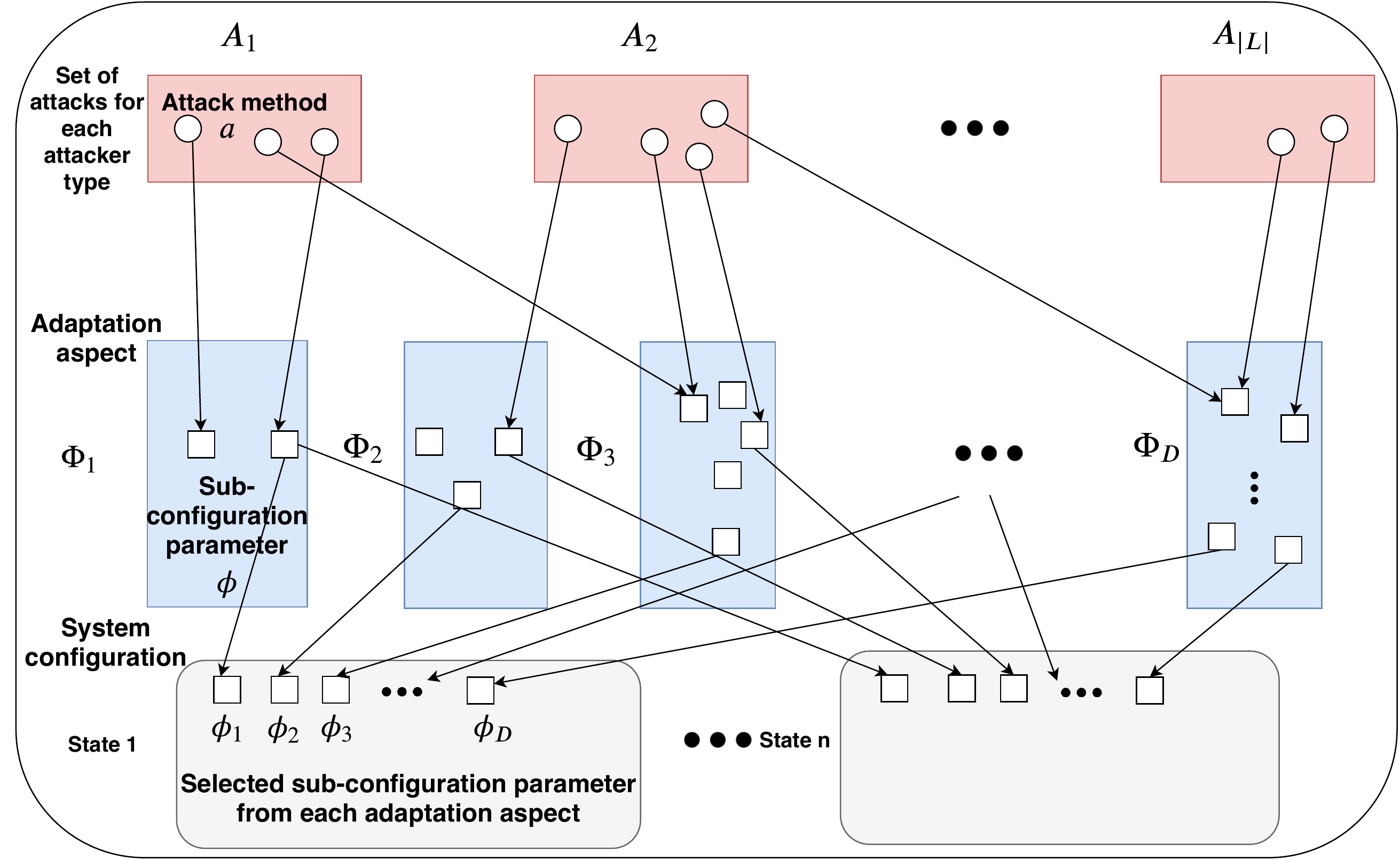}
\caption{An illustration of the relationship between attacker types, attack methods, adaptation aspects, sub-configuration parameters, and system configurations. }
\label{fig:relation}
\end{figure}

\subsubsection{Attacking Time}
If an attacker with the attacker type $l$ chooses an attack method $a \in A_{l}$ to attack the system in state $j$, then the time needed for him to compromise the system is a random variable $\xi_{a,j}$ that is drawn  from a given distribution $\Xi_{a,j}$. If $a$ is not targeting a vulnerability of any sub configuration in state $j$, $\xi_{a,j} = +\infty$. The attacker only gains benefits when he has successfully compromised the system. A system is considered to be compromised as long as the attacker compromises one of the sub configurations.
In this work, $L$, $\{A_l\}$, and $\{\Xi_{a,j}\}$ are assumed to be common knowledge.

\subsection{The Defender}
We consider a defender that proactively switches among different system configurations to increase the attacker's cost of mounting successful attacks. In practice, each migration incurs a cost to the defender. Frequent migrations often bring high cost while infrequent migrations make the system vulnerable to potential attacks. To make MTD feasible for deployment, the defender must also find the optimal time period to defend. An optimal MTD solution for the defender is to compute a strategy that simultaneously determines where to migrate and when to migrate.
\subsubsection{Migration Cost}
During the migration process, the defender updates the system to retake maintain the control of the system and pay some updating cost,
then selects and shifts the system from the current configuration state $i \in S$ to the next valid system configuration state $j \in S$ with a cost $m_{ij}$. We can consider that the migration is implemented instantaneously. This is without loss of generality because one can always assume that during the migration the system is still at state $i$. 
If the defender decides to stay at the current state $i$, then the cost is $m_{ii}$. 
Since there is a cost for updating the system, it requires that $m_{ii} >0$ for all $i,j \in S$. Let $M$ be the matrix of the migration cost between any states $i, j \in S$, we have $M = [m_{ij}]_{n  \times n}$. 
It is crucial to note that the migration cost $m_{ij}$ may depend on both the {\it source} configuration state $i$ and the {\it destination} configuration state $j$. 

\subsubsection{Defending Period}
The defender not only needs to determine which configuration state to use but also should decide when to move. If the defender stays in a configuration state sufficiently long, the attacker is likely to compromise the system even if the defender shifts eventually.  Thus, the defender needs to pick the defending period judiciously.
Let $t_{k}$ denote the time instance that the $k$th migration happens, the $k$th defending period is calculated by 
$t_{k} - t_{k-1}$. Intuitively, a defending period $\tau$ should be greater than zero, since the system needs some time to prepare the system migration or the updating. The defending period cannot be infinitely large because the system must switch the configurations periodically. Otherwise, it loses the benefits of moving target defense as the system will be compromised eventually and stay compromised after that. 
Therefore, it is natural to require that the defending period length $\tau$ has a lower bound $\underline{\tau} \in \mathbb{R}_{> 0}$ and an upper bound $\overline{\tau} \in \mathbb{R}_{> 0}$.  That is, $\tau \in [\underline{\tau}, \overline{\tau}]$. Note that the unit time cannot be infinitely small in practice. This allows us to discretize time for analytic convenience.

\subsubsection{Defender's Strategy}
In our model, the defender adopts a stationary strategy $\mathbf{u}$ that simultaneously decides where to move (depending on both the current state at $k$-th period and the next state at ($k+1$)-th period) and when to move (depending on the current state at $k$-th period). Let $p_{ij}$ denote the probability that the defender moves from state $i$ to state $j$, then the transition probabilities between any two states in $S$ can be represented by a transition matrix  $P=[p_{ij}]_{n\times n}$, where $\sum_{j\in S} p_{ij}=1$. When the system configuration state is $i$, the defender uses a mixed strategy $\mathbf{p}_i = (p_{ij})_{j \in S}$ to determine the next state to switch to.

In contrast, the defender needs to determine the next defending period $\tau_{i} \in [\underline{\tau}, \overline{\tau}]$ when she moves the system from current state  $i$ to the next state. The defending period matrix is thus $ \boldsymbol\tau = [\tau_{1}, \tau_{2}, \dots, \tau_{n}]$ where $i \in S$. In our model, we assume that the next ($k+1$)-th defending period only depends on the current system state $i$ at $k$-th period. There are two considerations: first, the transition probabilities $\mathbf{p}_{i}$ have already captured the differences between different destination states; second, it allows the defender's problem to be modeled as a semi-Markov decision process (defined in Section~\ref{sec:smdp}).
The defender's stationary strategy is denoted by $\mathbf{u} = (P, \boldsymbol\tau) = [u(i)]_{i \in S}$ where $u(i)=(\mathbf{p}_i,\tau_i)$ is the defender's action when the system's current configuration state is $i$. See Figure~\ref{fig:msgmtd} for an illustration of the system model. Note that the subscript $i$ in $\tau_i$ means the period length depends on state $i$ at $k$-th period. The time period $\tau_i$ is indifferent to the next state $j$ at ($k+1$)-th period for any $j\in S$.

\subsection{Markov Stackelberg Game}
In our paper, we model moving target defense as a Markov Stackelberg game (MSG) where the defender is the {\em leader} and the attacker is the {\em follower}. At the beginning the defender commits to a stationary strategy that is announced to the attacker. The attacker's goal is to maximize his expected reward by constantly attacking the system's vulnerabilities. The defender's objective is to minimize the total loss due to attacks plus the total migration cost.

The MSG model is an extension of the Stackelberg Security Game (SSG) model that has been widely adopted in security domains~\cite{sinha2018stackelberg}. A key advantage of using the Stackelberg game model is that it enables the defender acting first to commit to a mixed strategy while the attacker selects his best response after observing the defender's strategy~\cite{korzhyk2011stackelberg}. This advantage allows the defender to implement defense strategies prior to potential attacks~\cite{li2019optimal}. In MTD, the defender proactively switches the system between different configurations to increase the attacker's uncertainty. It is thus natural to model the defender as the leader and the attacker as the follower.

\begin{figure}[t]
\centering
 \includegraphics[width=.85\linewidth]{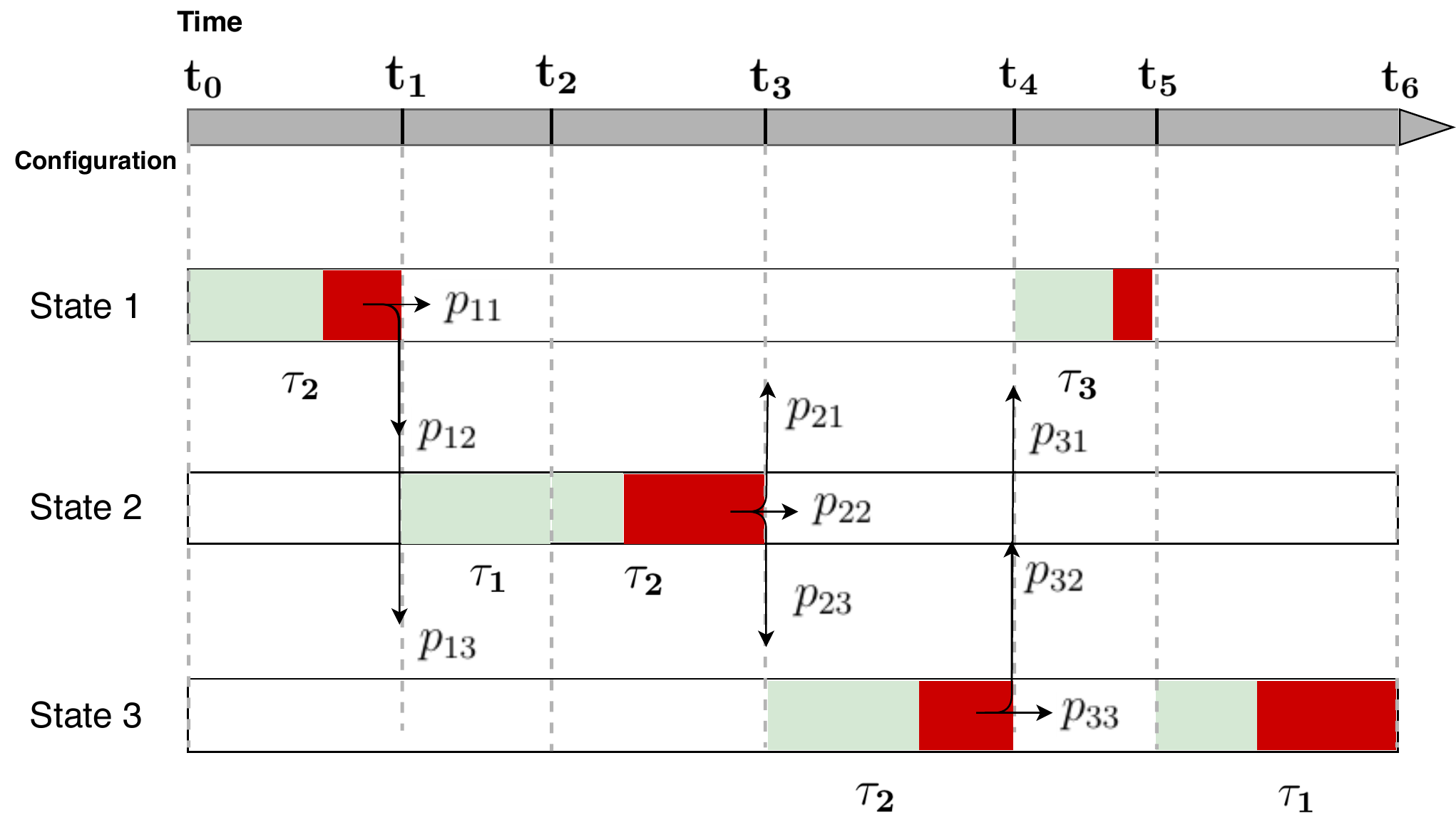}
\caption{An illustration of the Markov Stackelberg game model. The first defending period (between $t_0$ and $t_1$) depends on the initial state (assuming State $2$ is the initial configuration state). 
A light green block represents the time period when the system is protected while a dark red block denotes the time period when the system is compromised. Here $p_{ij}$ is the probability that the  defender moves from configuration state $i$ to $j$, and $\tau_{i} $ is the length of the current defending period when the previous configuration is $i$. }
\label{fig:msgmtd}
\end{figure}

\subsubsection{Information Structure}
In our model, we consider a stealthy and persistent attacker that learns the defender's stationary strategy ${\bf u}$.
Note that even if the defender's strategy is not announced initially, the attacker will learn the system states (e.g., through probing) and the defender's strategy eventually due to the stealthy and persistent nature of the attacks. Thus, it is without loss of generality to assume that the defender announces her stationary strategy to the attacker before the game starts. 
We further assume that the attacker always learns the previous state of the ($k-1$)-th period at the beginning of the $k$-th period no matter the attack is successful or not in the previous ($k-1$)-th period.
Hence, at the beginning of each stage, the only uncertainty to the attacker is the current configuration state (it knows the previous state and from that, it also knows when the current defending period will end as the length of which is determined by the previous state). Note that is a worse-case scenario from the defender's perspective. 

In many security domains, it is often difficult for the defender to obtain real-time feedback about whether the system is compromised or not because the attacker is likely to be stealthy. The defender, however, may have some prior knowledge about the attacker type distribution~\cite{paruchuri2008playing,sengupta2017game}. Thus, we assume that the defender has a prior belief $\mathbf{\pi}=\{\pi_l\}$ on the probability distribution of the attacker type $l$.

\subsubsection{Attacker's Optimization Problem}
 
The attacker chooses an attack method at the beginning of each stage to maximize its long-term reward. In our model, the defender adopts a stationary strategy that is known to the attacker. Thus, the long-run optimality can be achieved if the attacker always follows a best response in each stage. Hence, we consider a myopic attacker that aims to maximize his benefits by always using a best response to the defender's strategy according to its knowledge on the previous system configuration state. 

Consider an arbitrary stage and assume that the previous configuration state at the ($k-1$)-th period is $i$ and the current configuration state at the $k$-th period is $j$. For an attack $a$ to be successful in this stage, the time $\xi_{a,j}$ required for the attacker when using attack method $a$ targeting state $j$ should be less than the length of the defending period $\tau_{i}$.
Let $R^l_{a,j}$ denote the attacker's benefit per unit of time when the system is compromised, which is jointly determined by the attacker's type $l$, its chosen attack method $a$, and the state $j$. The reward that the attacker receives is then $(\tau_i-\xi_{a,j})^+R^l_{a,j}$ where $(x)^+ \triangleq \max(0,x)$.

The optimization problem for the attacker with type $l$ is to maximize his expected reward per defending period by choosing an attack method $a$ from his attack space $A_{l}$:
\begin{equation}
\max_{a\in A_l}\sum_{j\in S}p_{ij}\mathbb{E}[(\tau_i-\xi_{a,j})^+]R^l_{a,j}
\end{equation}
where $\tau_i$ and $\{p_{ij}\}$ are the defending period length and the transition probabilities given by the defender under state $i$, respectively.

\subsubsection{Defender's Optimization Problem}\label{sec:smdp}
We consider a defender that constantly migrates among the system configurations  in order to 
minimize the total loss due to attacks plus the total migration cost. We use an average cost semi-Markov decision process (SMDP) 
to model the defender's optimization problem. The SMDP model considers a defender that aims to minimize her long-term defense cost in an infinite time horizon using spatial-temporal decision making. 

 Let $C^l_{a,j}$ denote the unit time loss for the defender under state $j$ due to an attack $a$ launched by a type $l$ attacker. 
 The defender's action at state $i$ is $u(i)=(\mathbf{p}_i,\tau_i)$ and her expected single period cost is:
\begin{align}
    c(i, u(i)) =  \sum_{l \in L}\pi_l\left(\sum_{j \in S} p_{ij}\mathbb{E}[(\tau_i-\xi_{a^l,j})^+]C^l_{a^l,j}\right) + \sum_{j \in S} p_{ij}m_{ij} \nonumber \\
    s.t.\quad a^l = \arg\max_{a\in A_l}\left(\sum_{j\in S} p_{ij}\mathbb{E}[(\tau_i-\xi_{a,j})^+]R^l_{a,j}\right), \quad \forall l\in L
\end{align}
where the first part in the objective function is the expected attacking loss 
and the second part 
is the expected migration cost. 

The game starts at $t_{0} = 0$ and let $s_0$ be the initial state that is randomly selected from the state space $S$.  
The defender adopts a stationary policy $u$ where $u(s_k)=(\mathbf{p}_{s_k},\mathbf{\tau}_{s_k})$ for each state $s_k$, which generates an expected cost $c(s_k,u(s_k))$ that includes potential loss from compromises and migrations. Given the initial state $s_{0}$, the defender's long-term average cost is defined as:  
\begin{equation}
    z(s_0,u(s_0)) = \liminf_{N\rightarrow\infty}\frac{\sum^{N-1}_{k=0} c(s_k,u(s_k))}{\sum^{N-1}_{k=0}{\tau_{s_k}}}
\end{equation}

The goal of the defender is to commit to a stationary policy $\mathbf{u^*}=[u^*(i)]_n$ that minimizes the time-average cost for any initial state.
$z(i,u^*(i)) = \inf_{\mathbf{u}}z(i,u(i))$, $\forall i \in S$
For each $u(i)=(\mathbf{p}_i,\tau_i)$, we have $p_{ij}\in [0,1]$ for all $j$, $\sum_jp_{ij}=1$, and $\tau_i\in [\underline{\tau},\overline{\tau}]$. Thus, the action space $U(i)$ for every $u(i)$ is $[0,1]^n\times [\underline{\tau},\overline{\tau}]$, which is a continuous space.
We assume that $c(i,u(i))$ is continuous over $U(i)$. 
The defender's optimization problem corresponds to finding a strong Stackelberg equilibrium~\cite{breton1988sequential} where the defender commits to an optimal strategy assuming that the attacker will choose the best response to the defender's strategy and break ties in favor of the defender. This is a common assumption in Stackelberg security game literature.

\subsubsection{Challenges of Computing Optimal MTD Strategies}
The are two main challenges for the defender to compute the optimal strategies. First,
when an arbitrary transition matrix $P$ is allowed, the Markov chain associated with the given $P$ may have a complicated chain structure. 
The optimal solution may not exist and standard methods such as Value Iteration (VI) and Policy Iteration (PI)~\cite{puterman2014markov} may never converge when applied to an average-cost SMDP with a continuous action space. Second, a bilevel optimization problem needs to be solved in each iteration of VI or PI, which is challenging due to the infinite action space and the coupling of spatial and temporal decisions.

\subsection{MSG versus SSG}

Our Markov Stackelberg game model extends the classic Stackelberg Security Game (SSG)~\cite{paruchuri2008playing} in important ways.
In the classic SSG model, there is a set of targets and a defender has a limited amount of resources to protect them. The defender serves as the leader and commits to a mixed strategy. The attacker observes the defender's strategy (but not her action) and then responds accordingly. Thus, the SSG model is essentially a one-shot game. The SSG model has been extended to Bayesian Stackelberg Game (BSG) model to capture multiple attack types where the defender knows the distribution of attack types a prior as we assumed. In a recent work~\cite{sinha2018stackelberg}, the BSG model has been used to model moving target defense where only the spatial decision is considered and the defender commits to a vector $[p_j]_n$ where $p_j$ is the probability of moving to configuration state $j$ in the next stage, which is independent of the current configuration state. 

We note that the BSG model in~\cite{sinha2018stackelberg} is a special case of our MSG model. Specifically, let $\tau_{i} = 1$ for all $i \in S$ and $\xi^l_{a^l,j} = 0$ for all $a^{l} \in A^{l}, j \in S$. We further assume 
that $p_{ij} = p_j$ for all $j \in S$. Then the SMDP becomes an MDP with state independent transition probabilities. Since each row of the transition matrix is the same, the stationary distribution of the corresponding Markov chain is just $[p_j]_n$. Therefore, the average-cost SMDP reduces to the following one-stage optimization problem to the defender:
\begin{align}
    \min_{\mathbf{p}}\sum_l\pi_l\Big(\sum_j p_jC^l_{a^l,j}\Big) + \sum_{i,j}p_ip_{j}m_{ij} \nonumber\\
    s.t.\quad a^l = \arg\max_{a\in A_l}\Big(\sum_j p_jR^l_{a,j}\Big)\quad \forall l\in L
    \label{eq:defenderoptimization}
\end{align}
This is exactly the BSG model for MTD in~\cite{sinha2018stackelberg}. In the BSG variant, the transition probabilities depend on only the destination state. It corresponds to a special case in our model when the transition probabilities in each row are the same.
This simplified MTD strategy is optimal only if the migration cost depends on the destination state only but not the source state. 

Our MSG model enables the defender to handle the complex scenarios when the migration cost is both source and destination dependent. It also takes the defending period into account when in computing the optimal defense strategy. This consideration is useful because a stealthy and persistent attacker will compromise the system eventually if the system stays in a state longer than the corresponding attacking time. 

\section{Optimal Moving Target Defense}
This section presents an efficient solution to the defender's optimization problem in spatial-temporal moving target defense.  We first show that the original average-cost SMDP problem can be transformed into a discrete time Markov decision process (DTMDP) problem using a data transformation method. We then introduce a 
value iteration algorithm to solve the DTMDP problem and prove that the algorithm converges to a nearly optimal MTD policy. The algorithm involves solving a bilevel optimization problem in each iteration, which can be formulated as a mixed integer quadratic program. 

\subsection{Assumptions}
Before we present our solution, we first make two assumptions. 

\begin{assumption}
The transition probability matrix $P$ can be arbitrarily chosen by the defender. 
\label{asmp:transition}
\end{assumption}

\begin{assumption}
Given $\underline{\tau}>0$ as the lower bound of the defending period length, the defender's cost per unit time $c(i,u(i))/\tau_i$ is continuous and bounded over $U(i)$ for each $i$.
\label{asmp:bound}
\end{assumption}

Both assumptions are reasonable and can be easily satisfied. 
Assumption~\ref{asmp:transition} implies an important structure property of the SMDP as formally defined below. 

\begin{definition} [Communicating MDP~\cite{puterman2014markov}]
For every pair of states $i$ and $j$ in $S$, there exists a deterministic stationary policy $\mathbf{u}$ under which $j$ is accessible from $i$, that is, $\Pr(s_k=j|s_0=i, \mathbf{u})>0$ for some $k\geq 1$.    
\label{def:communicationmdp}
\end{definition}

It is easy to check that the SMDP in our problem is communicating under Assumption~\ref{asmp:transition}.
It is known that for a communicating MDP, the optimal average cost is a constant, independent of the initial state~\cite{puterman2014markov}. This property significantly simplifies the algorithm design and analysis as we discuss below. Assumption~\ref{asmp:bound} is used in establishing the convergence of the value iteration algorithm under the continuous action space (see Section~\ref{sec:analysis}).  

\subsection{Data Transformation}
Solving the defender's optimization problem requires the algorithm to simultaneously determine the optimal transition probabilities 
and the optimal defending periods. 
The average-cost SMDP problem with continuous action space is known to be difficult to solve~\cite{jianyong2004average}.  Fortunately, one can apply the data transformation method introduced by \citeauthor{schweitzer1971iterative}~\cite{schweitzer1971iterative} to transform the average-cost SMDP problem into a discrete-time average Markov decision process (DTMDP) problem. The DTMDP has a simpler structure than the SMDP with the same state space $S$, and action sets $U(i)=[u(i)]_{n}$, where $u(i)=(\mathbf{p}_i,\tau_i)$. 
The defender's per-stage cost $c(i, u(i))$ is converted to 
\begin{equation}
\tilde{c}(i,u(i)) = \frac{c(i,u(i))}{\tau_i}
\label{eq:tildeciui}
\end{equation}

\ignore{
\begin{equation}
\tilde{c}(i,u(i)) =\frac{\sum_l\pi_l\max_{a\in A_l}(\sum_{j}p_{ij} E[(\tau_i-\xi_{a,j})^+]R^l_{a,j} ) + \sum_{j} p_{ij}m_{ij}}{\tau_i}
\label{eq:tildeciui}
\end{equation}
}

Further, the transition probability from state $i$ to state $j$ for the DTMDP is 
\begin{equation}
\tilde{p}_{ij} ( u(i)) = \gamma \frac{p_{ij}-\delta_{ij}}{\tau_i} + \delta_{ij} 
\label{eq:tildepij}
\end{equation}
where 
$\delta_{ij}$ denotes the Kronecker delta (i.e., $\delta_{ii}=1$ and $\delta_{ij}=0$ for all $j\neq i$) and $\gamma$ is a parameter that satisfies $0 < \gamma < \underline{\tau} \leq \frac{\tau_i}{1-p_{ii}}$, where $\underline{\tau}$ is the lower bound of the defending period length. Let $\tilde{P}(\mathbf{u})=[\tilde{p}_{ij}(u(i))]_{n\times n}$ denote the transition probability matrix of the DTMDP and $\mathbf{\tilde{c}}(\mathbf{u})=[\tilde{c}(i,u(i))]_n$ the defender's per-stage cost across all the states. 
If the system starts from the initial state $s_{0} \in S$, then the long-term average cost becomes
\begin{align}
\tilde{z}(s_0,u(s_0)) &= \limsup_{N\rightarrow\infty}\frac{1}{N}\sum^{N-1}_{k=0} \tilde{c}(s_k,u(s_k)) 
\end{align}

The above data transformation has some nice properties as summarized below.
\begin{theorem}[Theorems 5.2 and 5.3 of ~\cite{jianyong2004average}]
Suppose an SMDP is transformed into a DTMDP using the above method. We have
\begin{enumerate}
    \item If SMDP is communicating, then DTMDP is also communicating.
    \item If SMDP is communicating, then a stationary optimal policy for DTMDP is also optimal for SMDP.
\end{enumerate}
\label{thm:datatransformation}
 \end{theorem}
 
Theorem~\ref{thm:datatransformation} indicates that the transformed DTMDP also has a contant optimal cost and further, to find a stationary optimal policy for the SMDP in our problem, it suffices to find a stationary optimal policy for the transformed DTMDP.

\subsection{Value Iteration Algorithm}
We adapt the {\em Value Iteration} (VI) algorithm~\cite{bertsekas1995dynamic,bather1973optimal} to solve the defender's problem and prove that the algorithm converges to a nearly optimal MTD policy. 
Before presenting the algorithm and the theoretical analysis, we introduce additional notations. 

\subsubsection{Additional Notations}
Let $V$ be any vector in $\mathbb{R}^n$. We define the mapping $F: \mathbb{R}^n \rightarrow \mathbb{R}^n$ as:
\begin{equation}
F(V) = \min_{\mathbf{u}}{[\tilde{\mathbf{c}}(\mathbf{u})+\tilde{P}(\mathbf{u})V]} \label{eq:F}
\end{equation}

\noindent where the minimization is applied to each state $i$ separately. For any vector $\mathbf{x}=(x_1,x_2,\dots,x_n)\in \mathbb{R}^n$, let $L(\mathbf{x}) = \min_{i=1,...,n}x_i$ and $H(\mathbf{x}) = \max_{i=1,...,n}x_i$.
Let $\|\cdot\|$ denotes the span seminorm defined as follows:
\begin{equation}
    \norm{\mathbf{x}}= H(\mathbf{x})-L(\mathbf{x})
\end{equation}
It is easy to check that $\|\cdot\|$ satisfies the triangle inequality, that is, $\norm{\mathbf{x}-\mathbf{y}}\leq \norm{\mathbf{x}}+\norm{\mathbf{y}}$ for any $\mathbf{x},\mathbf{y}\in \mathbb{R}^n$.
Further, $\norm{\mathbf{x}-\mathbf{y}}=0$ if and only if there exists a scalar $\lambda$ such that
$x-y = \lambda \mathbf{e}$ where $\mathbf{e}$ is 
an vector of all ones.
Thus, there is a vector $V$ such that $\norm{F(V)-V}=0$ ($V$ is called a fixed point of $F(\cdot)$) if and only if there is a scalar $\lambda$ such that the following optimality equation is satisfied:
\begin{equation}
\lambda \mathbf{e} + V = \min_{\mathbf{u}}{[\tilde{\mathbf{c}}(\mathbf{u})+\tilde{P}(\mathbf{u})V]}
\label{eq:optimalityequation}
\end{equation}
An important result in MDP theory~\cite{puterman2014markov,bertsekas1995dynamic} is that the stationary policy $\mathbf{u}$ that attains the minimum in the optimality equation~\eqref{eq:optimalityequation} is optimal and $\lambda$ gives the optimal long-term average cost.  

\ignore{We consider the relative value iteration where $V^{0} \in \mathbb{R}^{n}$,
\begin{equation}
W = V - L(V)\mathbf{e}
\end{equation}
If $L(V) = 0$, then $W$ and $y = x - L(x)\mathbf{e}$, then $\norm{y} = H(y)$ for any $x \in \mathbb{R}^n $. If $\{V\}$ is bounded, it follows that the $\{W\}$ is also bounded. The sequence $\{V\}$ is generated by $\kappa_{0} = \frac{1}{2}, V^{t} = F(\kappa_{t-1}V^{t-1})$, $\kappa_{t} = \frac{t}{t+1}$, where $t = 1, 2, \dots$.
}

\subsubsection{The VI Algorithm}
The VI algorithm (See Algorithm~\ref{alg:vi}) maintains a vector $V^t \in \mathbb{R}^n$. The algorithm starts with an arbitrary $V^{0}$ (line 1) and a carefully chosen sequence $\{\kappa_t\}$ (line 2) that ensures every limit point of $\{V^{t}\}$ is a fixed point of $F(V)$ (See Section~\ref{sec:analysis}). In each iteration, $V^t$ is updated by solving a policy improvement step ($V^t=\textbf{PImp}(S, V^{t-1},\underline{\tau}, \overline{\tau}, M, C, R, \pi, L,  \{A_l\}, \delta, \kappa_{t-1})$) in line 5.

In each policy improvement step (Algorithm ~\ref{alg:pi}), instead of finding the optimal $\mathbf{p}_i$ and $\tau_i$ together for each state $i$, which is a challenging problem, we discretize $[\underline{\tau},\overline{\tau}]$ and search for $\tau_i$ with a step size $\delta$ (line 3).
This approximation is reasonable since in practice the unit time cannot be infinitely small. Note that the smaller $\delta$ is, the closer $\mathbf{\tau}^*$ (line 7) is to the optimal one. Also note that the optimization problem in line 4-5 is actually a bilevel problem, which will be discussed in detail in Section~\ref{sec:bilevel}.

Under Assumptions~\ref{asmp:transition} and~\ref{asmp:bound}, the algorithm ~\ref{alg:vi} stops in a finite number of iterations (lines 6-8) and 
is able to find a near-optimal policy $P^{*}, \tau^{*}=\arg \textbf{PImp}(S, V^{t-1},\underline{\tau}, \overline{\tau}, M, C, R, \pi, L,  \{A_l\}, \delta)$ 
(formally proved in Section~\ref{sec:analysis}). 
In practice, a near-optimal solution is sufficient because it can be expensive or even unrealistic to obtain the exact minimum average cost in a large-scale MDP. 
The algorithm itself, however, still attains the optimal solution if the number of iterations goes to infinity (and $\delta$ approaches 0).
\begin{algorithm}[t]
\caption{Value Iteration algorithm for the MTD game}
\begin{algorithmic}[1]
\Statex \textbf{Input:} $S, n, \epsilon>0,\underline{\tau}, \overline{\tau}, M, C, R, \pi, L,  \{A_l\}, \delta>0.$
\Statex \textbf{Output:} $P^{*}, \tau^{*}$
\State $V^{0}\in \mathbb{R}^n;$
\State $\kappa_0=\frac{1}{2},$ $\kappa_t=\frac{t}{t+1}$ for $t=1,2,\dots.$
\Repeat
\State  $t=t+1$;
\State $V^t=\textbf{PImp}(S, V^{t-1},\underline{\tau}, \overline{\tau}, M, C, R, \pi, L,  \{A_l\}, \delta, \kappa_{t-1})$
\State $\overline{V}=\max_{i\in S}{|V^{t}(i)-V^{t-1}(i)|};$
\State $\underline{V}=\min_{i\in S}{|V^{t}(i)-V^{t-1}(i)|};$
\Until {$\overline{V}-\underline{V} < \epsilon$}
\State $P^{*}, \tau^{*}=\arg\textbf{PImp}(S, V^{t-1},\underline{\tau}, \overline{\tau}, M, C, R, \pi, L,  \{A_l\}, \delta)$
\end{algorithmic}
\label{alg:vi}
\end{algorithm}

\begin{algorithm}[t]
\caption{Policy Improvement (\textbf{PImp})}
\begin{algorithmic}[1]
\Statex \textbf{Input:} $S, V_0, \underline{\tau}, \overline{\tau}, M, C, R, \pi, L, \{A_l\}, \delta, \kappa$
\Statex \textbf{Output:} $V_t$
\For{$i\in S$}
\State $v=+\infty;$
\For{$\tau=\underline{\tau};\tau\leq\overline{\tau};\tau=\tau+\delta$}
\State $V_t(i)=\min_{\mathbf{p}_i}[\tilde{c}(i, \tau, M, C, R, \pi, L,  \{A_l\})$
\State $+ \kappa\sum_{j\in S} \tilde{p}_{ij}(\mathbf{p}_i, \tau)V_0(j)]$
\If{$V_t(i) < v$}
\State $ v = V_t(i);$
\EndIf
\EndFor
\State  $V_t(i)=v$;
\EndFor
\end{algorithmic}
\label{alg:pi}
\end{algorithm}

\subsubsection{Theoretical Analysis}\label{sec:analysis}
Our VI algorithm is adapted from the work due to \citeauthor{bertsekas1995dynamic}~\cite{bertsekas1995dynamic} and  \citeauthor{bather1973optimal}~\cite{bather1973optimal} that originally addresses the average cost MDP problem with a finite state space and an infinite action space. However, their proofs do not directly apply to our algorithm because they either consider the transition probabilities as the only decision variables~\cite{bather1973optimal} or involve the use of randomized controls~\cite{bertsekas1995dynamic}. 
In contrast, our strategy includes both the probability transition matrix and the (deterministic) defending periods. 

For a given stationary policy $\mathbf{u}$ with transition matrix $\tilde{P}$, let $\tilde{P}^*$ denote the Cesaro limit given by $\tilde{P}^{*} = \lim_{N\to\infty}  \{I + \tilde{P} + \tilde{P}^{2} + \dots + \tilde{P}^{N-1}\}/N$. Then the average cost associated with $\mathbf{u}$ can be represented as $\tilde{P}^*\tilde{c}(\mathbf{u})$\cite{puterman2014markov}. The policy $\mathbf{u}$ is called {\it $\epsilon$-optimal} if $\tilde{P}^*\tilde{c}(\mathbf{u}) \leq \lambda + \epsilon \mathbf{e}$ where $\lambda$ is the optimal cost vector. 
In practice, it is often expensive or even unrealistic to compute an exact optimal policy and an $\epsilon$-optimal policy might be good enough. 
\ignore{
Instead, it is reasonable to seek a near-optimal policy that has some guarantee. In our paper, we consider an $\epsilon$-optimal policy (See Definition~\ref{def:epsilonoptimal}).
\begin{definition}
Given $\epsilon> 0$, a stationary policy $\mathbf{u} = (P, \tau)$ is $\epsilon$-optimal if  $\tilde{P}(P, \tau)\mathbf{\tilde{c}(P, \tau)} \leq \lambda + \epsilon \mathbf{e}$
where $\lambda = \min_{u \in U} \tilde{P}(P, \tau)\tilde{c}(P, \tau)$. 
\label{def:epsilonoptimal}
\end{definition}
}
Our main results can be summarized as follows. 

\begin{theorem}
Under Assumptions~\ref{asmp:transition} and~\ref{asmp:bound}, we have
\begin{enumerate}
\item The DTMDP problem (thus the SMDP problem too) has an optimal stationary policy;
\item The sequence of policies in Algorithm~\ref{alg:vi} eventually leads to an $\epsilon$-optimal policy. 
\end{enumerate}
\label{thm:riv}
\end{theorem}
\noindent{\it Proof Sketch:} The first part can be proved using the similar techniques in the proofs of Theorem 2.4 of~\cite{bather1973optimal} and Proposition 5.2 of~\cite{bertsekas1995dynamic}. The main idea is to show that (1) $\{\|V^t\|\}$ is bounded thus the vector sequence $\{V^t\}$ must have a limit point; (2) every limit point of $\{V^t\}$ is a fixed point of $F(\cdot)$, thus leading to an optimal solution. The second part follows from Theorem 6.1 and Corollary 6.2 of ~\cite{bather1973optimal}. The main idea is to show that (1) if $\|F(V^{t})-V^{t}\| \leq \epsilon$, then the corresponding policy is $\epsilon$-optimal; (2) $\lim_{t \to \infty} \|F(V^{t})-V^{t}\|$ = 0 (again using the boundedness of $\{\|V^t\|\}$) so that $\|F(V^{t})-V^{t}\| \leq \epsilon$ holds eventually. Note that this condition is exactly the stopping condition $\overline{V} - \underline{V} < \epsilon$ in Algorithm~\ref{alg:vi} (line 8).

The proof of Theorem~\ref{thm:riv} relies on the key property that the vector sequence $\{\|V^{t}\|\}$ generated by Algorithm~\ref{alg:vi} is bounded. Due to the coupling of spatial and temporal decisions in our problem, the techniques in~\cite{bather1973optimal,bertsekas1995dynamic} cannot be directly applied to prove this fact. Below we provide a detailed proof on this result by adapting the techniques in~\cite{bather1973optimal,bertsekas1995dynamic}. We prove that $\{\norm{V^t}\}$ is bounded. 

\ignore{
\begin{lemma}[Lemma 2.1 of~\cite{bather1973optimal}]
For any vectors $\mathbf{x},\mathbf{y}\in \mathbb{R}^n$ and constants $a,b \in \mathbb{R}$, we have
\begin{enumerate}
    \item $H(F(\mathbf{x})-F(\mathbf{y}))\leq H(\mathbf{x}-\mathbf{y})$.
    \item $\norm{F(\mathbf{x})-F(\mathbf{y})}\leq \norm{\mathbf{x}-\mathbf{y}}$.
    \item $\norm{F(a\mathbf{x})-F(b\mathbf{y})}\leq|a|\norm{\mathbf{x}-\mathbf{y}}+|a-b|\norm{\mathbf{y}}$.
\end{enumerate}
\end{lemma}
}

\begin{lemma}
Let Assumptions~\ref{asmp:transition} and~\ref{asmp:bound} hold, and $\{\kappa_t\}$ be a nondecreasing sequence with $\kappa_t\in [0,1]$ for each $t$. Consider a sequence $\{V^t\}$ where 
\begin{equation*}
    V^{t+1}=F(\kappa_tV^t)=\min_{P,\mathbf{\tau}}{[\tilde{\mathbf{c}}(P,\mathbf{\tau})+\kappa_t\tilde{P}(P,\mathbf{\tau})V^t]},
\end{equation*}
then $\{\norm{V^t}\}$ is bounded.
\label{lmm:bounded}
\end{lemma}

\begin{proof}
Without loss of generality, we assume $V^0=0\mathbf{e}$.
Since $\tilde{c}(P,\mathbf{\tau})$ is bounded according to Assumption~\ref{asmp:bound}, there exists a constant $\beta$ such that $0\mathbf{e} \leq \tilde{\mathbf{c}}(P,\mathbf{\tau})\leq \beta \mathbf{e}$.
Then using the fact that $\kappa_t$ is nondecreasing, we can shown that $\{V^t\}$ is nondecreasing by induction. 
For a communicating system, for each pair of states $i$ and $j$, there exists a stationary policy 
$\mathbf{u}_{ij}$ such that $j$ is accessible from $i$. Now we combine the transition probability matrices from these policies to form a new transition probability matrix
\begin{equation}
    Q=\frac{1}{n^2}\sum^n_{i=1}\sum^n_{j=1}\tilde{P}(u_{ij}) \label{eq:Q}
\end{equation}
The key observation is that the matrix $Q$ is also a valid transition probability matrix, that is, there exists a policy $\mathbf{u}$ such that $Q = \tilde{P}(\mathbf{u})$. This is because the defender can choose arbitrary migration probabilities and the defending period lengths as her strategy. Formally, for an arbitrary transition probability matrix $Q = [q_{ij}]$ given by~\eqref{eq:Q}, we may solve the data-transformation equations to identify the corresponding policy $\mathbf{u} = (P,\tau)$. 
\begin{align}
    q_{ij} &= \frac{\gamma p_{ij}}{\tau_i}, \forall j \neq i \label{eq:solve_q_1} \\
    q_{ii} & = \frac{\gamma(p_{ii}-1)}{\tau_i}+1 \label{eq:solve_q_2}
\end{align}

\noindent To see this is always possible, we first observe that for any probability matrix $Q = [q_{ij}]$ given by~\eqref{eq:Q}, we must have $q_{ij} \in [0,\gamma/\underline{\tau}]$ for $j \neq i$ and $q_{ii} \in [1-\gamma/\underline{\tau},1]$ using the definition of the data transformation. Define a policy $\mathbf{u}=(P, \boldsymbol\tau)$ where $ \boldsymbol\tau = \underline{\tau}\mathbf{e}$ and $P = [p_{ij}]$ is given by
\begin{align*}
    p_{ij} &= \frac{\underline{\tau} q_{ij}}{\gamma}, \forall j \neq i\\
    p_{ii} & = 1- \frac{\underline{\tau}}{\gamma}(1-q_{ii}) 
\end{align*}
\ignore{
\begin{align}
    q_{ij\neq i} &= \frac{(\underline{\tau}-\varepsilon)p_{ij}}{\tau_i}\\
    q_{ii} & = \frac{(\underline{\tau}-\varepsilon)(p_{ii}-1)}{\tau_i}+1 
\end{align}
Given $\mathbf{u}=(P,\underline{\tau}\mathbf{e})$, we have 
\begin{align}
    q_{ij\neq i} &= p_{ij\neq i} (1 - \frac{\varepsilon}{\underline{\tau}})\\
    q_{ii} & = p_{ii} (1-\frac{\epsilon}{\underline{\tau}})+\frac{\varepsilon}{\underline{\tau}}
\end{align}
From the second equation, we know that $q_{ii} \geq \frac{\varepsilon}{\underline{\tau}}$.  Since for each row $i$ the constraint $\sum_jq_{ij}=1$ holds, we have $q_{ij\neq j}\leq 1-\frac{\varepsilon}{\underline{\tau}}$ . Given $p_{ij}\in[0,1]$ and $\sum_jp_{ij}=1$, it follows that any $q_{ij}$ can be constructed by a $p_{ij}$ and any row vector $\mathbf{q}_i$ can be constructed by a $\mathbf{p}_i$.}
Using the above mentioned properties of $q_{ij}$ and the definition of $\gamma$, it is easy to verify that $P$ is a probability transition matrix and $\mathbf{u}=(P, \boldsymbol\tau)$ defined above satisfies equations~\eqref{eq:solve_q_1} and~\eqref{eq:solve_q_2}.

From the definition of $Q$, every state is accessible from every other state. Define $T_{ij}$ as the expected number of transitions to reach $j$ from $i$, we have
\begin{align*}
    T_{ij}= 1 + \sum_{k\neq j}q_{ik}T_{kj}, \quad i,j\in S,i\neq j
\end{align*}
We can then show by induction that
\begin{equation*}
V^t(i) \leq \beta T_{ij}+V^t(j), \quad i,j\in S,i\neq j, t=0,1,\dots  
\end{equation*}
\ignore{Consider the base case $k=0$, $V^0=0\mathbf{e}$, the inequality holds. We use $k$ in the induction step,
\begin{equation}
    V_i^{k+1} \leq \tilde{\mathbf{c}}(\tilde{P},\mathbf{\tau}) + \kappa_kQV^k \leq \beta\mathbf{e}+QV_j^k
\end{equation}
For all $i,j\in S$, 
\begin{align}
    V_i^{k+1} &\leq \beta+\sum^n_{l=1} q_{il}V^k_l \\
    &= \beta +\sum_{l\neq j}q_{jl}V^k_l+q_{ij}V^k_j \\
    &= \beta(1+\sum_{l\neq j}q_{il}T_{lj})+\sum_{l\neq j}q_{jl}V^k_l+q_{ij}V^k_j \\
    &=\beta T_{ij}+V^k_j
\end{align}
Since we know $\{V^k\}$ is nondecreasing, hence
\begin{equation}
V_i^{k+1} \leq \beta T_{ij}+V_j^{k+1}, \quad i,j\in S,i\neq j, k=0,1,\dots  
\end{equation}
which completes the induction.
}
It follows that
\begin{equation*}
    \norm{V^t}\leq \max\{\beta T_{ij}| i,j\in S,i\neq j\}
\end{equation*}
Thus, $ \norm{V^t}$ is bounded. This concludes the proof.
\end{proof}

\subsection{Bilevel Optimization Problem}
\label{sec:bilevel}
To compute the optimal defense strategy with Algorithm~\ref{alg:vi}, we need to solve the following optimization problem for a given scalar $\tau$ and a vector $V^{t-1}$ (line 5, 9 in Algorithm~\ref{alg:vi}):
\begin{equation*}
    V^t(i)=\min_{\mathbf{p}_i}{[\tilde{c}(i,\mathbf{p}_i, \tau)+ \kappa_{t-1}\sum_{j\in S} \tilde{p}_{ij}(\mathbf{p}_i, \tau)V^{t-1}(j)]}
\end{equation*}
Substitute $\tilde{c}(i,\mathbf{p}_i,\tau)$ and $\tilde{p}_{ij}(\mathbf{p}_i,\tau)$ by their definitions in Equations~\eqref{eq:tildeciui} and~\eqref{eq:tildepij}, 
\ignore{
\begin{align*}
    V^t(i) &= \min_{\mathbf{p}_i}\Bigg[\frac{\sum_l\pi_l\max_{a\in A_l}\sum_{j}p_{ij} E[(\tau_i-\xi_{a,j})^+]R^l_{a,j} }{\tau_i}  \\ &+ \frac{\sum_{j} p_{ij}[m_{ij}+\gamma\kappa_{t-1} W^{t-1}(j)]}{\tau_i}+\kappa_{t-1}(1-\frac{\gamma}{\tau_i})W^{t-1}(i)\Bigg] 
\end{align*}
}
and denote $w^l_{j,a} \triangleq E[(\tau_i-\xi^l_{a,j})^+]$ and $\theta_j \triangleq m_{ij}+\gamma \kappa_{t-1} V^{t-1}(j)$ to simplify the notation. 
The defender's optimization problem then simplifies to (with the constant terms in the objective function omitted): 
\begin{align*}
    \min_{\mathbf{p}_i}\sum_l\pi_l\left(\sum_j p_{ij}w^l_{j,a^l}C^l_{a^l,j}\right)+\sum_jp_{ij}\theta_j\\
    s.t.\quad p_{ij} \in [0,1], \forall j \in S;  \sum_{j\in S} p_{ij}=1; \\
    \quad a^l = \arg\max_{a\in A_l}\left(\sum_j p_{ij}w^l_{j,a}R^l_{a,j}\right)\quad \forall l\in L.
\end{align*}
Using the similar technique for solving Bayesian Stackelberg games~\cite{paruchuri2008playing}, this bilevel optimization problem 
can be modeled as a Mixed Integer Quadratic Program (MIQP):
\begin{align*}
\min_{\mathbf{p}_i,\mathbf{n},\mathbf{v}}\sum_{j\in S}\sum_{l\in L}\sum_ {a\in A_l}\pi_lw^l_{j,a}C^l_{a,j}p_{ij}n^l_a+\sum_jp_{ij}\theta_j\\
s.t.\quad
\sum_{j\in S}p_{ij}=1\\
\sum_{a\in A_l} n^l_a=1, \quad \forall l\in L\\
0\leq v^l-\sum_{j}p_{ij}w^l_{j,a}R^l_{a,j}\leq(1-n^l_a)B, \quad \forall a \in A_l,l\in L \numberthis \label{bigM}\\
p_{ij}\in[0,1], \ \ n^l_a=\{0,1\}, \ \ v^l\in \mathbb{R}, \ \ \forall j \in S, a \in A_l, l \in L
\end{align*}
where the binary variable $n^l_a = 1$ if and only if $a \in A_l$ is the best action for the type $l$ attacker. This is ensured by constraint~\eqref{bigM} where $v^l$ is an upper bound on the attacker's reward and $B$ is a large positive number.

\subsection{Relative Value Iteration}

\begin{algorithm}[t]
\caption{Relative Value Iteration algorithm for the MTD game}
\begin{algorithmic}[1]
\Statex \textbf{Input:} $S, n, \epsilon>0, \underline{\tau}, \overline{\tau}, M, C, R, \pi, L,  \{A_l\}, \delta>0.$ 
\Statex \textbf{Output:} $P^{*}, \tau^{*}$
\State $V^{0}\in \mathbb{R}^n$;
\State $W^0=V^0-V^0(s)\mathbf{e}$;
\State $\kappa_0=\frac{1}{2},$ $\kappa_t=\frac{t}{t+1}$ for $t=1,2,\dots.$
\Repeat
\State  $t=t+1$;
\State $V^t=\textbf{PImp}(S, W^{t-1},\underline{\tau}, \overline{\tau}, M, C, R, \pi, L,  \{A_l\}, \delta, \kappa_{t-1})$
\State $W^t=V^t-V^t(s)\mathbf{e}$
\State $\overline{V}=\max_{i\in S}{|V^{t}(i)-V^{t-1}(i)|};$
\State $\underline{V}=\min_{i\in S}{|V^{t}(i)-V^{t-1}(i)|};$
\Until $\overline{V}-\underline{V} < \epsilon$
\State $P^{*}, \tau^{*}=\arg\textbf{PImp}(S, W^{t-1},\underline{\tau}, \overline{\tau}, M, C, R, \pi, L,  \{A_l\}, \delta, \kappa_{t-1})$
\end{algorithmic}
\label{alg:rvi}
\end{algorithm}

The VI algorithm can be slow in practice due to the large number of iterations needed to converge and the complexity of solving multiple MIQP problems in each iteration.  
To obtain a more efficient solution, we introduce a Relative Value Iteration (RVI) algorithm (see Algorithm~\ref{alg:rvi}). The RVI algorithm maintains both a state-value vector $V^t$ similar to the VI algorithm as well as a relative state-value vector $W^t$ defined as $W^t=V^t-V^t(s)\mathbf{e}$ where $s$ is a fixed state (line 2). The RVI algorithm again starts with an arbitrary vector $V^0$. In each iteration $t$, $V^t$ and $W^t$ are updated as $V^t= F(W^{t-1})$ (line 6) where $F$ is defined in Equation~\ref{eq:F} and $W^{t+1}=V^t-V^t(s)\mathbf{e}$ (lines 7). Under Assumptions \ref{asmp:transition} and \ref{asmp:bound}, we can show that the sequence $\{W^t\}$ converges to a vector $W^*$ satisfying the optimality equation $F(W^*)(s)\mathbf{e}+W^*=F(W^*)$ similar to Equation~\eqref{eq:optimalityequation}, where $F(W^*)(s)$ gives the optimal average cost.
Further, Theorem~\ref{thm:riv} still holds by substituting $\norm{V^{t+1}-V^t}$ with $\norm{W^{t+1}-W^t}$. 

Theoretically, the RVI algorithm does not help improve the efficiency of the VI algorithm since $W^t$ and $V^t$ only differ by a multiple of the vector of all ones and the bilevel optimization problems involved in both algorithms are mathematically equivalent. In terms of converge rates, $\norm{W^{t+1}-W^{t}}$ has the same bound as $\norm{V^{t+1}-V^{t}}$, as determined by the sequence $\{\kappa_t\}$ (line 2 in Algorithm \ref{alg:vi} and line 3 in Algorithm \ref{alg:rvi}) and the upper bound of $\norm{V^t} = \norm{W^t}$. In practice, however, the RVI algorithms often requires smaller number of iterations for the same $\epsilon$. In addition, if we introduce the assumption that $p_{is}\geq\rho>0$ for all $i\in S$ (where $\rho$ can be arbitrarily small) to further restrict the Markov chain structure and set $\kappa_t=1$ for all $t$, the rate of convergence can be significantly improved~\cite{bertsekas1995dynamic} for both the VI and the RVI algorithms.

\section{Empirical Evaluation}
We conducted numerical simulations using the real data from the National Vulnerability Database (NVD)~\cite{o2009national} to demonstrate the advantages of using MSG for spatial-temporal MTD. In particular, we derived the key attack/defense parameters from the Common Vulnerabilities Exposure (CVE) scores~\cite{mell2006common} in NVD, which has been widely used to describe the weakness of a system with respect to certain risk levels. We used data samples in NVD with CVE scores ranging from January 2013 to August 2016. As in~\cite{sengupta2017game}, the Base Scores (BS) and Impact Scores (IS) were used to represent the attacker's reward (per unit time) and the defender's cost (per unit time) respectively. Further, we used the Exploitability Scores (ES) to estimate the distribution of attack time. 

We conducted two groups of experiments\footnote{Code is available at \href{https://github.com/HengerLi/SPT-MTD}{https://github.com/HengerLi/SPT-MTD}.}: the spatial decision setting and the joint spatial-temporal decision setting. We compared the MSG method with two benchmarks: the Bayesian Stackelberg Game (BSG) model~\cite{sengupta2017game} and the Uniform Random Strategy (URS) ~\cite{sengupta2017game}.
We used the Gurobi solver (academic version 8.1.1) for the MIQP problems in BSG and MSG. All the experiments were run on the same 24-core 3.0GHz Linux machine with 128GB RAM.

\subsection{Spatial Decision}

\subsubsection{Benchmarks and Settings:}
In the spatial decision setting, the defender periodically moves in unit time length and the attacker instantaneously compromises the system when he chooses the correct configuration. We compared the MSG model with the original BSG and URS models in~\cite{sengupta2017game}.
In BSG, the defender determines the next configuration according to a fixed transition probability vector $[p_j]_n$ that is independent of the current configuration. In URS, the defender selects the next configuration uniform randomly.

For fair comparisons, we followed the same data generation method as used in the work by~\citeauthor{sengupta2017game}~\cite{sengupta2017game}. The system has four configurations $S=\{(PHP,MySQL)$, $(Python, MySQL)$, $(PHP, postgreSQL)$, $(Python, postgreSQL)\}$ with the switching cost shown in Figure \ref{cost_f}. The default migration cost matrix $M$ is $[[2,4,8,12]$, $[4,2,11,7]$, $[8,11,2,12]$, $[12,7,12,2]]$ (we added an updating cost of 2 to all the switching cost in~\cite{sengupta2017game}.) We varied the updating cost (see the numbers in the parentheses of the 
blue boxes in Figure~\ref{cost_f}) and evaluated the impact of different updating cost on the performance of defender's strategies.

In this experiment, we considered three attacker types: the Script Kiddie that could attack $Python$ and $PHP$, the Database Hacker that is able to attack $MySQL$ and $postgreSQL$ and the Mainstream Hacker that could attack all the techniques. The defender possesses a prior belief of $(0.15, 0.35, 0.5)$ on the three attacker types. The size each of their attack space is $34, 269$ and $48$, respectively. For BSG, the defender's optimization problem was directly solved with MIQP as in~\cite{sengupta2017game}. For MSG, the bi-level optimization problem was solved with MIQP for each configuration in every iteration of Algorithm~\ref{alg:vi} (with the convergence parameter $\epsilon = 0.1$). 
Since the simulated migration cost could not be directly compared with the attacking cost from NVD, we introduced a parameter $\alpha$ to adjust the ratio between the attacking cost and the migration cost. That is, instead of the $m_{ij}$ shown in Figure~\ref{cost_f}, we used $\alpha m_{ij}$ as the migration cost from state $i$ to state $j$. As $\alpha$ increases, the migration cost has a larger impact on the defender's decisions.

\begin{figure}[t]
\centering
 \includegraphics[width=.95\linewidth]{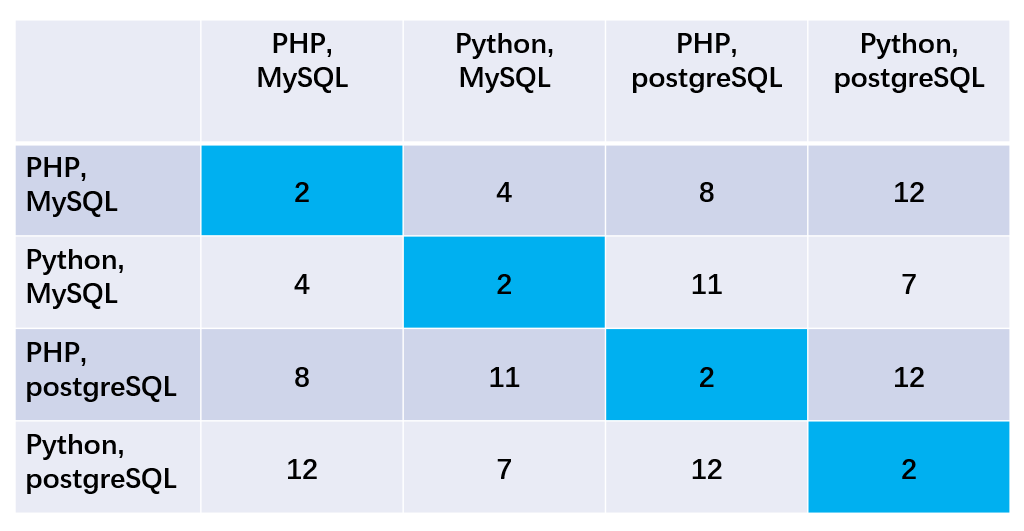}
\caption{The migration cost in the MTD system. Each row represents a source configuration and each column represents a destination configuration. The updating costs are shown in the blue boxes.}
\label{cost_f}
\end{figure}

\begin{figure}[t]
\centering
\includegraphics[width=.95\linewidth]{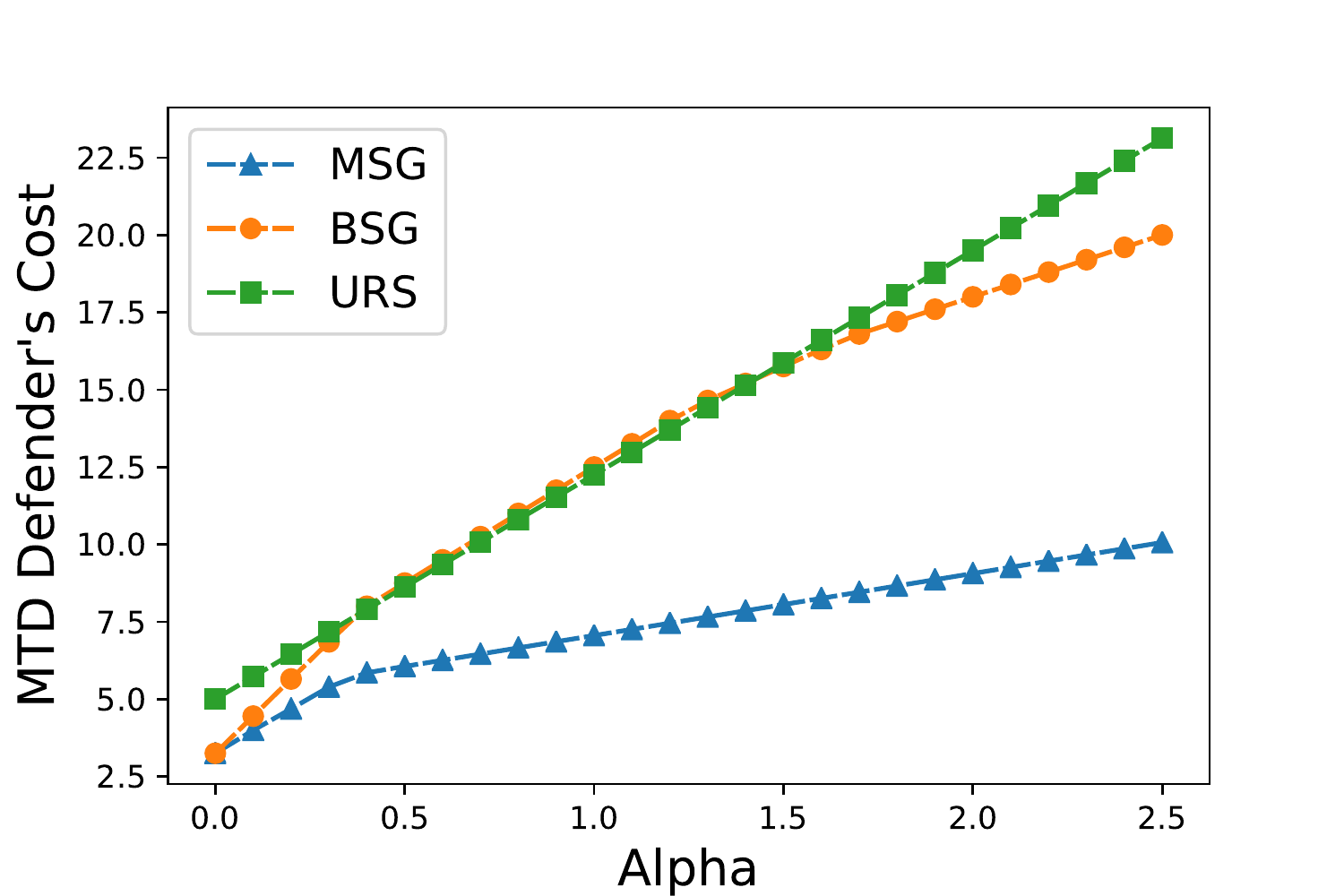}
\caption{A comparison of the defender's cost in the three policies with spatial decisions only - MSG ($\epsilon=0.1$), BSG and URS with unit defending period ($\tau_{i} = 1$ for all $i \in S$) and zero attacking time 
as the parameter $\alpha$ increases.}
\label{fig:MSGvsBSG}
\end{figure}

\begin{figure}[t]
\centering
 \includegraphics[width=.95\linewidth]{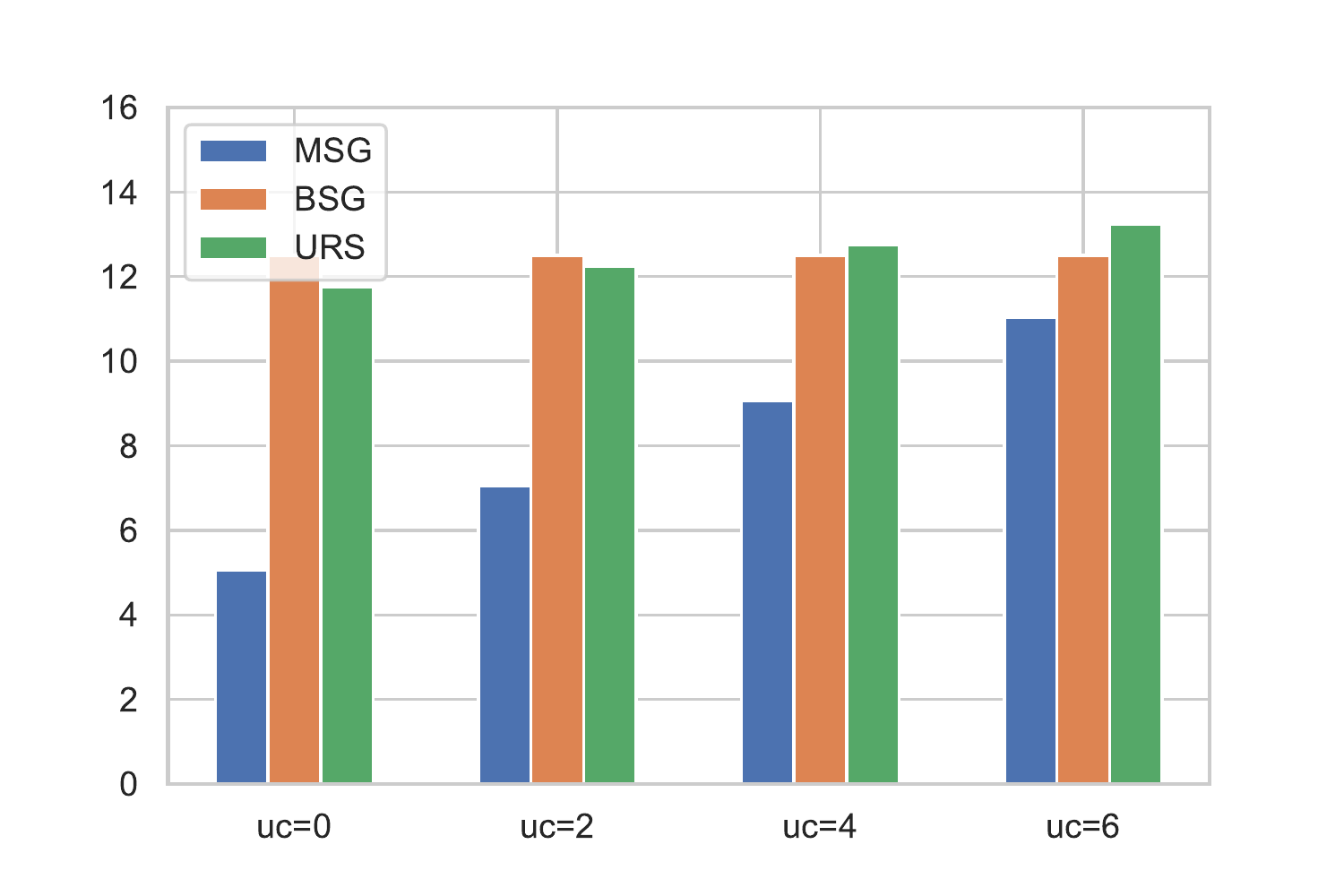}
\caption{The defender's cost in the three policies with spatial decisions only - MSG ($\epsilon=0.1$), BSG and URS under four different updating cost $0,2,4,6$ with $\alpha=1$. The four settings change the default updating cost 2 (in blue boxes of Figure \ref{cost_f}) to  $0,2,4,6$ while keeping the other values unchanged.}
\label{cost}
\end{figure}

\subsubsection{Results:} 
We varied the value of $\alpha$ 
from $0$ to $2.5$ with an increment of $0.1$ and compared the defender's cost for the three policies: the MSG, the BSG, and the URS (See Figure~\ref{fig:MSGvsBSG}). All three models were restricted to unit length defending period and zero attacking time. 
That is $\tau_{i} = 1$ for all $i \in S$, and $\xi^l_{a^l,j} = 0$ for all $a^{l} \in A^{l}, j \in S$. 

Figure~\ref{fig:MSGvsBSG} shows that the defender's cost increases for all the three policies as the migration cost grows. However, the magnitude of increase differs in the three policies. In URS, the cost increases linearly due to the uniform random strategy $(0.25,0.25,0.25,0.25)$ used. In both MSG and  BSG, the defender's cost grows sub-linearly. However, the defender incurs substantially less cost in MSG. The reason is that although both MSG and BSG enable the defender to choose the respective optimal strategies, MSG allows the defender to vary her strategy according to different source configurations while in BSG the defender must choose the same strategy for all the source configurations. When $\alpha$ is small ($\alpha\in [0,0.3]$), BSG uses $(0, 0, 0.5, 0.5)$ as the defender's strategy, and MSG chooses almost the same strategy for each configuration (MSG could do better if we impose temporal decisions as shown below). 
When $\alpha \in [0.4, 1.5]$, BSG uses strategies $(0.25,0.25,0.25,0.25)$ and $(0.25, 0.5, 0, 0.25)$. 
In contrast, as $\alpha$ grows within $[0.4, 1.6]$, MSG chooses the most beneficial configuration $(PHP, postgreSQL)$ (in term of the attacking cost plus the updating cost) as the absorbing state. 
At the turning point $\alpha =0.6$ and $\alpha=1.2$,
the defender changes the spatial strategy at $(Python, MySQL)$ from $(0,0,0.5,0.5)$ to $(0,0,0,1)$ and the spatial strategy at $(PHP, MySQL)$ from $(0,0,0.5,0.5)$ to $(0.5,0,0,0.5)$ to trade uncertainty for less migration cost. These source-dependent adjustments can't be achieved by BSG.
the Markov chain structure with both the absorbing states (when the system moves in, it always stays there) and the transient states (once the system moves out, it never come back) cannot be achieved by BSG since the only way to stay in a configuration using BSG is to assign a probability of 1 to that configuration for all source configurations, which, however, removes any uncertainty to the attacker. This indicates that MSG can achieve a better trade-off between the migration cost and the loss from attacks, where the latter is determined by the uncertainty of the attacker.
When $\alpha$ is large, the overall migration cost is high. In BSG, the defender uses the strategy  $(0.5,0.5,0,0)$ (for $\alpha \in [1.6, 2.5]$) to move between the two configurations with relatively low migration cost (which is still higher than the updating cost). 
For $\alpha \in [1.7,2.5]$, MSG chooses the strategy (i.e. (0,1,0,0) instead of (0.5,0,0,0.5) at $(PHP,MySQL)$) to avoid the costly move from $(PHP,MySQL)$ to $(Python, postgreSQL)$ that has a migration cost of $12$, the defender chooses another path by firstly moving from $(PHP,MySQL)$ to $(Python, MySQL)$ and then to $(Python, postgreSQL)$. This again brings some uncertainty to the attacker.

We further compared the defender's cost under four different updating cost settings (see Figure~\ref{cost}). In URS, the defender's cost lineally increases as updating cost grows due to the uniform spatial strategy. In BSG, the defender's cost stay the same because of the inaccurate estimation of $p_ip_j$ using the piecewise linear McCormick envelopes \cite{sengupta2017game}. In contrast, MSG generates the accurate probabilities $p_{ij}$ corresponding to the configuration pair in each migration. We observe that as the updating cost increases, the gap of defender's cost between MSG and other two models decreases. The reason is that when the divergence between the moving cost and updating cost is flatted, the advantage of setting the most beneficial configuration as an absorbing state (i.e., avoiding unnecessary switching) become negligible. MSG has significant advantages over the other two methods when  the migration cost matrix (e.g. the updating cost is significantly smaller than moving cost) is unbalanced and the cost values vary from different source configurations.

\subsection{Joint Spatial-Temporal Decision}
\subsubsection{Benchmarks and Settings:} In the joint spatial-temporal decision setting, the defender needs to decide not only the next configuration to move to but also the length of each defending period $\tau$. In our experiments, $\tau$ is in the range of $[0.1, 2.6]$ with an increment parameter $\delta=0.1$. For MSG, the optimal $\tau_i$ for each configuration $i$ was obtained together with the spatial decisions using Algorithm~\ref{alg:vi}.
We extended the BSG and URS policies by incorporating the attacking times and the defending periods into the objectives of both the defender and the attacker (as we did in our MSG model), where a fixed defending period is used for all the configurations since both policies ignore the source configuration in each migration. To have a fair comparison with MSG, we searched for the optimal defending period for BSG and URS, respectively, by solving the following optimization problems: 

\begin{enumerate}
    \item Uniform Random Strategy (URS) with Temporal Decision:
    \begin{alignat*}{1}
        \min_{\tau}\frac{\sum_l\pi_l\Big(\sum_j \mathbb{E}[(\tau-\xi_{a,j})^+]C^l_{a^l,j}\Big)+\frac{\alpha}{n}\sum_{i,j}m_{ij}}{n\tau}
        \nonumber\\
        s.t.\quad a^l = \arg\max_{a\in A_l}\Big(\sum_j \frac{1}{n}\mathbb{E}[(\tau-\xi_{a,j})^+]R^l_{a,j}\Big)\quad \forall l\in L
    \end{alignat*}
    \item Bayesian Stackelberg Game (BSG) with Spatial-Temporal Decisions:
    \vspace{-1ex}
    \begin{alignat*}{2}
        \min_{\mathbf{p},\tau}\frac{\sum_l\pi_l\Big(\sum_j p_j\mathbb{E}[(\tau-\xi_{a,j})^+]C^l_{a^l,j}\Big) +\alpha \sum_{i,j}p_{i} p_{j} m_{ij}}{\tau} \nonumber\\
        s.t.\quad a^l = \arg\max_{a \in A_l}\Big(\sum_j p_j\mathbb{E}[(\tau-\xi_{a,j})^+]R^l_{a,j}\Big)\quad \forall l\in L
    \end{alignat*}
\end{enumerate}


We assigned a random attacking time for each attack that aims to compromise the system. The random attacking time $\xi_{a,j}$ was drawn from the exponential distribution $Exp({ES_a})$ (the mean attacking time is $1/ES_a$) when $a$ is targeting a vulnerability in state $j$ and $\xi_{a,j} = + \infty$ otherwise. 
Here, $ES_a$ refers to the {\em exploitability score} of the vulnerability targeted by attack method $a$. The $ES$ score of a vulnerability is a value between 0 and 10 and 
a higher $ES$ score means it is easier to exploit the vulnerability~\cite{mell2007complete}. 

We used the same migration cost setting as the spatial decision experiment with the default migration cost matrix $[[2,4,8,12]$, $[4,2,11,7]$, $[8,11,2,12]$, $[12,7,12,2]]$. For MSG, we set the convergence parameter $\epsilon$ as $0.1$. For each vulnerability, we generated 1000 samples from the corresponding exponential distribution and used their average as the attacking time. 

\begin{figure}[t]
\centering
 \includegraphics[width=.95\linewidth]{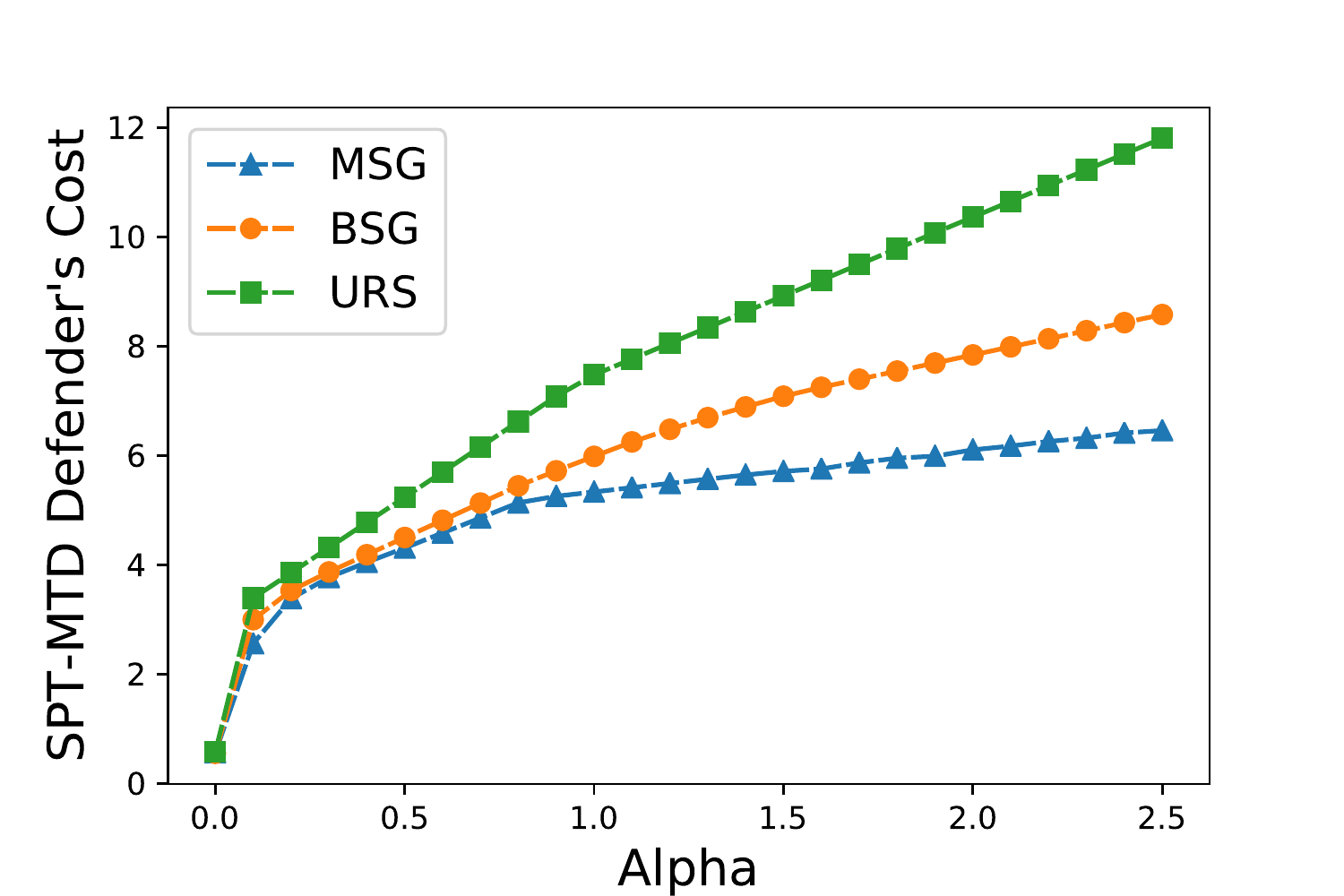}
\caption{A comparison of the defender's cost in the three spatial-temporal policies - MSG ($\epsilon=0.1$), BSG and URS as the parameter $\alpha$ increases. $\xi_{a,j}\sim Exp({ES_a})$.}
\label{MSG_tau}
\end{figure}

\begin{figure}[t]
\centering
 \includegraphics[width=.95\linewidth]{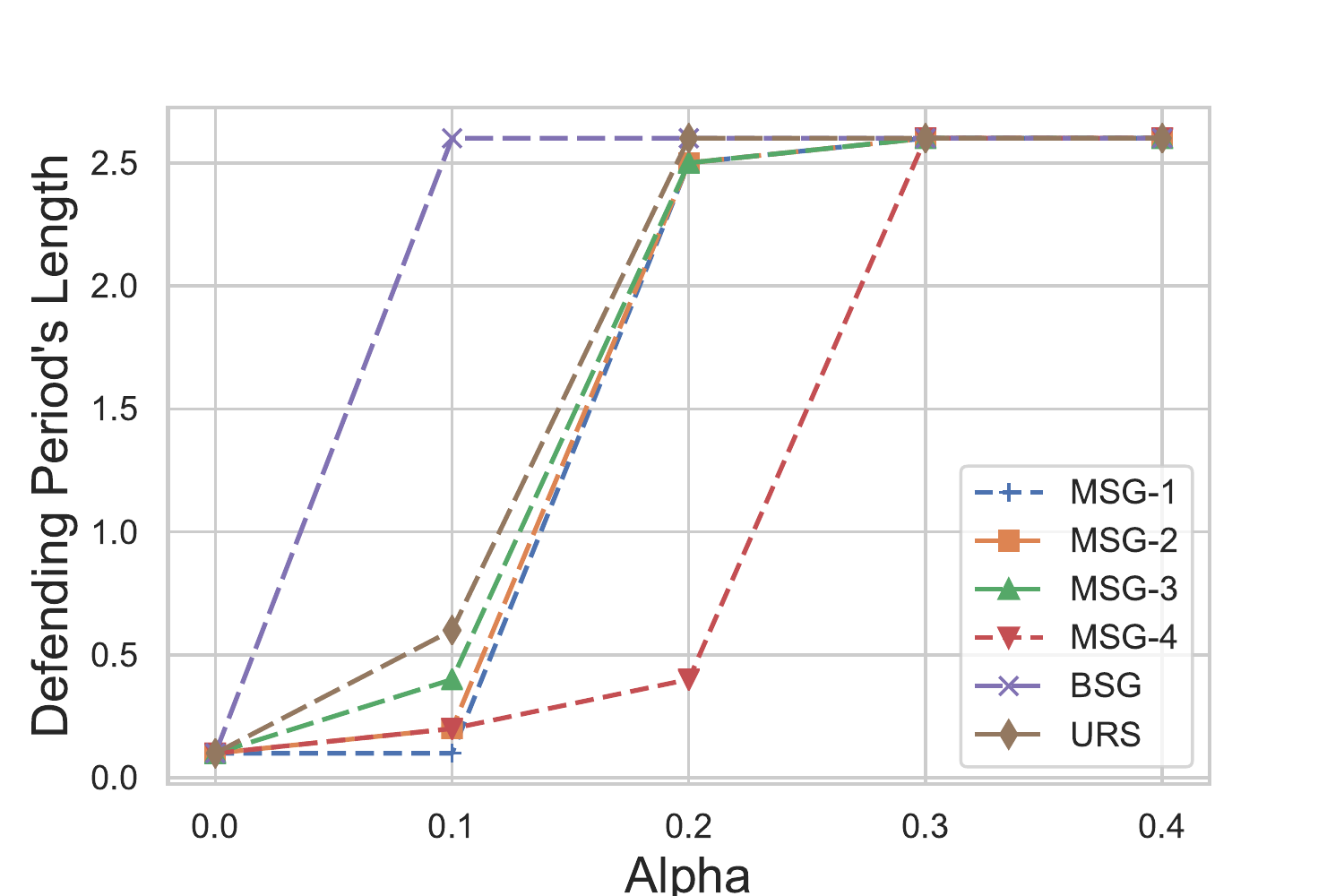}
\caption{A comparison of the defender's temporal decisions in the three policies when $\alpha \in [0, 0.4]$ with a step size of $0.1$ and $\tau \in [0, 2.6]$ with an increment $\delta=0.1$. For all the three policies, the optimal defending period of every configuration reaches $2.6$ when $\alpha \geq 0.3$. MSG-1, MSG-2, MSG-3 and MSG-4 represent the defender's temporal decisions in MSG for configurations $S=\{(PHP,MySQL)$, $(Python, MySQL)$, $(PHP, postgreSQL)$, and $(Python, postgreSQL)\}$, respectively.}
\label{tau}
\end{figure}

\begin{figure}[t]
\centering
 \includegraphics[width=.95\linewidth]{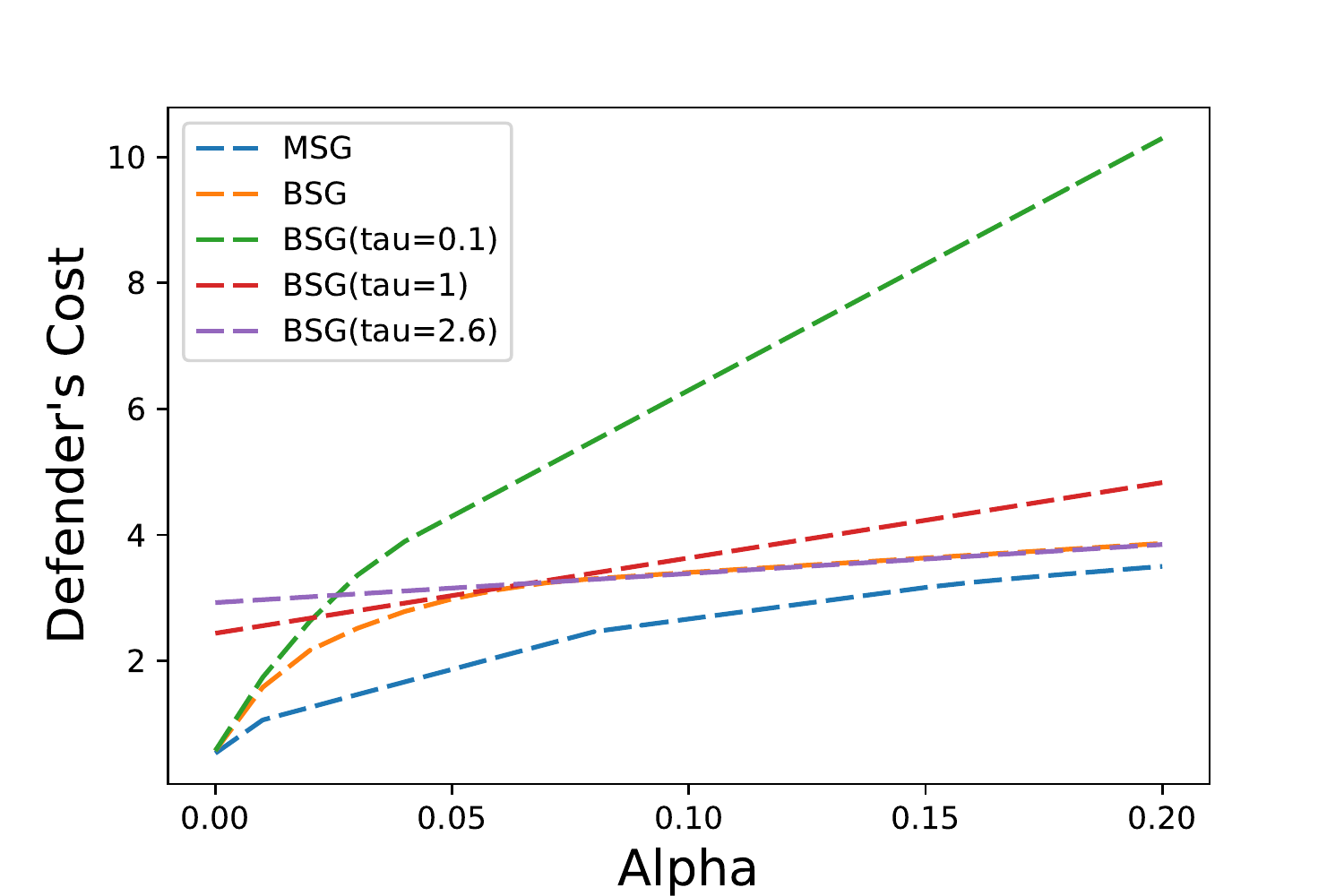}
\caption{A comparison of the defender's cost between MSG and BSG with different temporal decisions ($\tau=0.1, 1, 2.6$ and the optimal $\tau$). The cost parameter $\alpha \in [0, 0.2]$ has an increment of $0.01$. The defender's optimal spatial decision in BSG remains unchanged in these settings. }
\label{BSG_tau}
\end{figure}

\subsubsection{Results:}
When the defender is able to decide when to migrate, all three models produce lower cost than the respective models with a fixed unit defending period (See Figure~\ref{MSG_tau} and Figure ~\ref{fig:MSGvsBSG}). 
When $\alpha$ is small, the attacking cost has a major impact on the defender's cost. The shorter defending periods lead to more frequent switches that can efficiently increases the uncertainty of the attacker (Figure ~\ref{tau}), thus reduce the attacking cost in the end.    
As $\alpha$ grows, the migration cost becomes the major factor and the defender's spatial decision plays a major role on the cost since now all models choose the same $\tau=2.6$ to decrease the unit time migration cost.

After adding temporal decision, BSG and URS are capable to adjust the frequency of moving, thus, improve the performance significantly. Compare with them, the improvement of MSG is not that obvious 
since the defender with MSG has already achieved the adjustment in some levels by setting some absorbing states (zero movement).

We observed that for the spatial decisions, the defender in the MSG model tends to move to configurations with lower attacking cost, namely $(PHP, postgreSQL)$ and $(Python, postgreSQL)$, and never move out of them. That is, these two configurations are recurrent states and the other two configurations are transient states of the corresponding Markov chain. When $0.2<\alpha<0.9$, the defender tends to move between the two recurrent states to increase the attacker's uncertainty, while for larger $\alpha\geq0.9$, $(Python, postgreSQL)$ becomes an absorbing states to reduce the migration cost. In contrast, the defender under BSG chooses the strategy $(0,0,0.5,0.5)$ when $\alpha$ is small and adopts the uniform strategy  $(0.25,0.25,0.25,0.25)$ when $\alpha$ becomes large. The latter is because configurations with higher average incoming migration cost have lower attacking cost. Besides, the total cost (the attacking cost plus the average incoming migration cost) becomes more similar across different configurations when $\alpha$ becomes large.  
Essentially, MSG has advantages over BSG in all the scenarios because it allows a more refined trade-off between migration cost and attacking cost, where the latter is determined by the attacker's uncertainty. 

For the temporal decisions, MSG, BSG and URS always choose the maximum $\tau$ (i.e., $2.6$) when $\alpha>0.3$. This temporal strategy is reasonable when the migration cost weighs more than the attacking cost. In this scenario, increasing $\tau$ would decrease the frequency of migration. From Figure~\ref{BSG_tau}, we can see that the optimal temporal decision $\tau$ increases as the migration cost weighs more ($\alpha$ grows). Thus, another reason that MSG outperforms BSG for small $\alpha$ is because MSG enables the defender to choose a longer period length on the configurations that have lower attacking cost such as $(PHP,MySQL)$ (see Figure~\ref{tau}).

\section{Conclusions}
In this paper we consider a defender's optimal moving target problem in which both sequences of the system configurations and the timing of switching are important. We introduce a Markov Stackelberg Game framework to model the defender's spatial and temporal decision making that aims to minimize the losses caused by compromises of systems and the cost required for migration. We formulate the defender's optimization problem as an average-cost SMDP and transform the SMDP problem into an DTMDP problem that can be solved efficiently. We propose an optimal and efficient algorithm to compute the optimal defense policies through relative value iteration. Experimental results on real-world data demonstrate that our algorithm outperforms the state-of-the-art benchmarks for MTD. Our Markov Stackellberg Game model precisely captures a defender's spatial-temporal decision making in face of adaptive and sophisticated attackers. 

Our work opens up new avenues for future research. In our paper, we have considered the scenario when the defender has  prior information about the distribution of the attacker type. It is interesting to study the case when the distribution of the attacker type is unknown to the defender.  In our work, we have assumed that the attacker is myopic. One may consider the MTD problem when the attacker has bounded rationality~\cite{kar2015game}.

\section{Acknowledgment}
This work has been funded in part by NSF grant CNS-1816495. We thank the anonymous reviewers for their constructive comments. We would like to thank Sailik Sengupta from the Arizona State University for kindly providing their MIQP code using in BSG with the piecewise linear McCormick envelopes.


%
%


\bibliographystyle{ACM-Reference-Format}  
\bibliography{mtd}  

\end{document}